%% file: main.tex
\PassOptionsToPackage{usenames,dvipsnames,svgnames,table}{xcolor}
\documentclass[sigconf,natbib=false,authorversion,nonacm]{acmart}

\pdfoutput=1

\AtBeginDocument{%
  }

\copyrightyear{2024}
\acmYear{2024}
\setcopyright{acmlicensed}\acmConference[LADC 2024]{13th Latin-American Symposium on Dependable and Secure Computing}{November 26--29, 2024}{Recife, Brazil}
\acmBooktitle{13th Latin-American Symposium on Dependable and Secure Computing (LADC 2024), November 26--29, 2024, Recife, Brazil}
\acmDOI{10.1145/3697090.3697095}
\acmISBN{979-8-4007-1740-6/24/11}

\RequirePackage[
  datamodel=acmdatamodel,
  style=acmnumeric,
  ]{biblatex}

\usepackage{algorithm}
\usepackage[noend]{algpseudocode}
\usepackage{xspace}
\usepackage{multirow}
\usepackage{makecell}
\usepackage{enumitem}

\usepackage{xcolor}
\usepackage{tikz}
\usetikzlibrary{positioning,calc,arrows,decorations.pathreplacing,patterns}
\usepackage{color,soul}

\let\ensuremathorig=\ensuremath
\renewcommand{\ensuremath}[1]{\ensuremathorig{#1}\xspace}

\let\paragraphorig=\paragraph
\renewcommand{\paragraph}[1]{\paragraphorig{\bf{#1}}}

\newcommand{\plusone}{\ensuremath{\text{\small\texttt{+1}}}}

\newcommand{\any}{\ensuremath{\_}}

\newcommand{\sep}{\ensuremath{\ldotp}}

\newcommand{\gamcast}{\ensuremath{\text{GM-Send}}}
\newcommand{\gdeliver}{\ensuremath{\text{GM-Deliver}}}

\newcommand{\gbcast}[1]{\ensuremath{\text{GB-Send}_{#1}}}
\newcommand{\gbdeliver}[1]{\ensuremath{\text{GB-Deliver}_{#1}}}

\algblockdefx[action]{StartAction}{EndAction}[1]{%
  \setboolean{preIndent}{true}%
  \setboolean{effIndent}{true}%
  #1}{}

\makeatletter
\ifthenelse{\equal{\ALG@noend}{t}}%
  {\algtext*{EndAction}}
  {}%
\makeatother

\tikzset{fontscale/.style = {font=\relsize{#1}}}

\newcommand{\YES}{{\color{OliveGreen} $\surd$\xspace}}
\newcommand{\SOME}{{\color{Orange} $\surd$\xspace}}
\newcommand{\NO}{{\color{BrickRed} $\times$\xspace}}

\newcommand\callout[2][]{\tikz[overlay]\node[fill=#1,inner sep=.75pt, anchor=text, rectangle] {#2};\phantom{#2}}

\newcommand{\labfigure}[1]{\label{fig:#1}}
\newcommand{\reffigure}[1]{Figure~\ref{fig:#1}}

\newcommand{\labtab}[1]{\label{tab:#1}}
\newcommand{\reftab}[1]{Table~\ref{tab:#1}}

\newcommand{\labsect}[1]{\label{sect:#1}}
\newcommand{\refsect}[1]{Section~\ref{sect:#1}}
\newcommand{\refsecttwo}[2]{Sections~\ref{sect:#1}~and~~\ref{sect:#2}}

\newcommand{\labappendix}[1]{\label{appendix:#1}}

\newcommand{\labalg}[1]{\label{alg:#1}}
\newcommand{\refalg}[1]{Algorithm~\ref{alg:#1}}

\newcommand{\labline}[1]{\label{line:alg:#1}}
\newcommand{\refline}[1]{line~\ref{line:alg:#1}}
\newcommand{\reflines}[2]{lines~\ref{line:alg:#1} to~\ref{line:alg:#2}}
\newcommand{\reflinestwo}[2]{lines~\ref{line:alg:#1} and~\ref{line:alg:#2}}

\newcommand{\labinv}[1]{\label{inv:#1}}
\newcommand{\refinv}[1]{Invariant~\ref{inv:#1}}

\newcommand{\lablem}[1]{\label{lem:#1}}
\newcommand{\reflem}[1]{Lemma~\ref{lem:#1}}

\newcommand{\labprop}[1]{\label{prop:#1}}
\newcommand{\refprop}[1]{Proposition~\ref{prop:#1}}

\newcommand{\labcor}[1]{\label{cor:#1}}

\newcommand{\labequation}[1]{\label{eq:#1}}
\newcommand{\refequation}[1]{Equation~(\ref{eq:#1})}

\newcommand{\union}{\ensuremath{\cup}}
\newcommand{\inter}{\ensuremath{\cap}}

\newcommand{\code}[1]{\ensuremath{\text{\texttt{#1}}}}
\newcommand{\msg}[1]{\code{#1}}

\newcommand{\beginMsg}{\msg{Begin}}
\newcommand{\proposeMsg}{\msg{Propose}}
\newcommand{\deliverMsg}{\code{Deliver}}

\newcommand{\aread}{\code{read}}
\newcommand{\awrite}{\code{write}}
\newcommand{\acas}{\code{cas}}
\newcommand{\abegin}{\code{begin}}
\newcommand{\aend}{\code{end}}

\newcommand{\key}{\code{key}}
\newcommand{\isRead}{\code{isRead}}
\newcommand{\operations}{\code{ops}}

\newcommand{\Gr}{\mathcal{G}}

\newcommand{\msgBegin}{\code{Begin}}
\newcommand{\msgPropose}{\code{Propose}}
\newcommand{\msgDeliver}{\code{Deliver}}
\newcommand{\msgAdvance}{\code{Advance}}

\newcommand{\delOrder}{\ensuremath{\rightarrow}}
\newcommand{\delOrderOf}[1]{\ensuremath{\rightarrow_{#1}}}

\newcommand{\equaldef}{\ensuremath{\stackrel{\triangle}{=}}}

\newcommand{\tlaplus}{\ensuremath{{\text{TLA}^{+}}}}
\newcommand{\SAF}{\ensuremath{\text{\textsf{SAF}}}}
\newcommand{\SAFa}{\ensuremath{\text{\textsf{SAFa}}}}
\newcommand{\SAFb}{\ensuremath{\text{\textsf{SAFb}}}}
\newcommand{\LIV}{\ensuremath{\text{\textsf{LIV}}}}
\newcommand{\LIVa}{\ensuremath{\text{\textsf{LIVa}}}}
\newcommand{\LIVb}{\ensuremath{\text{\textsf{LIVb}}}}
\newcommand{\LIVc}{\ensuremath{\text{\textsf{LIVc}}}}

\newif\iflong
\longfalse
\longtrue % uncomment to compile long version

\newboolean{showcomments}
\setboolean{showcomments}{true}
\ifthenelse{\boolean{showcomments}}
{ \newcommand{\mynote}[3]{
    \fbox{\bfseries\sffamily\scriptsize#1}
    {\small$\blacktriangleright$\textsf{\emph{\color{#3}{#2}}}$\blacktriangleleft$}}}
{ \newcommand{\mynote}[3]{}}

\begin{document}

%%
%% The "title" command has an optional parameter,
%% allowing the author to define a "short title" to be used in page headers.

\iflong
\title[Generic Multicast (Extended Version)]{Generic Multicast\\(Extended Version)} 
\else
\title[Generic Multicast]{Generic Multicast}
\subtitle{Full research paper}
\fi

\author{Jos{\'e} Augusto Bolina}
\authornote{The main contributions of these authors were made while studying as MSc students at the Computer Science School of Federal University of Uberlândia, Brazil.}
\email{jbolina@redhat.com}
\orcid{0000-0001-6654-6442}
\affiliation{%
  \institution{Red Hat, Inc.}
  \country{Brazil}
}

\author{Pierre Sutra}
\email{pierre.sutra@telecom-sudparis.eu}
\orcid{0000-0002-0573-2572}
\affiliation{
    \institution{T{\'e}l{\'e}com SudParis}
    \institution{INRIA}
    \country{France}
}

\author{Douglas Antunes Rocha}
\authornotemark[1]
\email{douglasanr@ufu.br}
\orcid{0009-0009-7237-7484}
\affiliation{%
  \country{Brazil}
}

\author{Lasaro Camargos}
\email{lasaro@weilliptic.com}
\orcid{0000-0002-4162-6160}
\affiliation{%
  \institution{Weilliptic Inc.}
  \country{USA}
}
\affiliation{%
   \institution{Federal University of Uberlândia}
   \country{Brazil}}

%%
%% By default, the full list of authors will be used in the page
%% headers. Often, this list is too long, and will overlap
%% other information printed in the page headers. This command allows
%% the author to define a more concise list
%% of authors' names for this purpose.
\renewcommand{\shortauthors}{Bolina et al.}

\begin{abstract}
Communication primitives play a central role in modern computing.
They offer a panel of reliability and ordering guarantees for messages, enabling the implementation of complex distributed interactions.
In particular, atomic broadcast is a pivotal abstraction for implementing fault-tolerant distributed services.
This primitive allows disseminating messages across the system in a total order.
There are two group communication primitives closely related to atomic broadcast.
Atomic multicast permits targeting a subset of participants, possibly stricter than the whole system.
Generic broadcast leverages the semantics of messages to order them only where necessary (that is when they conflict).
In this paper, we propose to combine all these primitives into a single, more general one, called generic multicast.
We formally specify the guarantees offered by generic multicast and present efficient algorithms.
Compared to prior works, our solutions offer appealing properties in terms of time and space complexity.
In particular, when a run is conflict-free, that is no two messages conflict, a message is delivered after at most three message delays.
\end{abstract}

%%
%% The code below is generated by the tool at http://dl.acm.org/ccs.cfm.
%% Please copy and paste the code instead of the example below.
%%
\begin{CCSXML}
<ccs2012>
   <concept>
       <concept_id>10010520.10010575</concept_id>
       <concept_desc>Computer systems organization~Dependable and fault-tolerant systems and networks</concept_desc>
       <concept_significance>500</concept_significance>
       </concept>
   <concept>
       <concept_id>10003752.10003809.10010172</concept_id>
       <concept_desc>Theory of computation~Distributed algorithms</concept_desc>
       <concept_significance>500</concept_significance>
       </concept>
 </ccs2012>
\end{CCSXML}

\ccsdesc[500]{Computer systems organization~Dependable and fault-tolerant systems and networks}
\ccsdesc[500]{Theory of computation~Distributed algorithms}

%%
%% Keywords. The author(s) should pick words that accurately describe
%% the work being presented. Separate the keywords with commas.
\keywords{Consensus, Multicast, Broadcast, Generalized Consensus.}

%%
%% Update accordingly
\iflong
\else
\received{19 July 2024}
\received[accepted]{25 September 2024}
\fi

%%
%% This command processes the author and affiliation and title
%% information and builds the first part of the formatted document.
\maketitle

%%%%%%%%%%%%%%%%%%%%%%%%%%%%
%
% INTRODUCTION
%
%%%%%%%%%%%%%%%%%%%%%%%%%%%%
\section{Introduction}
\labsect{introduction}

Atomic broadcast is a fundamental building block of modern computing infrastructures.
This communication primitive offers strong properties in the ordering and delivery of messages.
Atomic broadcast finds usage in many storage systems, from file systems and relational databases to object stores and blockchains \cite{chandra2007paxos,hunt2010zookeeper,corbett2013spanner,ongaro2014search,CamaioniGMRVV24}.
It ensures the scalability, high availability, and fault tolerance of these distributed systems.
% atomic multicast
At its core, atomic broadcast guarantees the reliable delivery of messages in the same total order across the system.
Such guarantees generalize into two natural and related group communication primitives.

The first primitive is called \emph{atomic multicast}.
Atomic multicast allows sending a message to a (possibly stricter) subset of the processes in the system.
In this case, the ordering property becomes a partial order linking the pairs of messages having a joint destination (a common process).
Atomic multicast helps to distribute the load and better leverage workload parallelism.
It is used in geo-replicated and partially-replicated systems where data is stored in multiple partitions \cite{granola,BenzMPG14a}.

The second common generalization of atomic broadcast is \emph{generic broadcast} (aka., generalized consensus).
This primitive is introduced in the work of \textcite{pedone2002handling} and \textcite{lamport2005generalized}.
Here, the ordering property binds messages that conflict, that is, messages whose processing does not commute in the upper application layer.
Generic broadcast leverages the semantics of messages to expedite delivery and improve overall performance \cite{lamport2005generalized,pedone2002handling,epaxos}.

\paragraph{Contributions}

In this work, we propose to combine the two primitives, atomic multicast and generic broadcast, into a new primitive called \emph{generic multicast}.
As atomic multicast, generic multicast permits sending a message to a subset of the processes in the system.
Delivery is based on the semantics of messages, as in generic broadcast.
This paper defines generic multicast and proposes two new solutions.

With more details, we claim the following contributions:
\emph{(i)} the definition of generic multicast for crash-prone distributed systems,
\emph{(ii)} a base solution that extends the timestamping approach of Skeen \cite{birman1987reliable},
\emph{(iii)} building upon this, an understandable fully-fledged fault-tolerant generic multicast algorithm.

Our algorithms extend well-established mechanisms proposed in the literature (\textit{e.g.}, \cite{birman1987reliable, fritzke1998fault, schiper2007optimal, ahmed2016convoy}).
Our base (non fault-tolerant) solution delivers a message in three communication steps when no two messages conflict.
Tolerating failures attains the same lower bound, for a total of~3 communication steps.
To the best of our knowledge, this is the first time this delivery latency has been attained in a conflict-free scenario (see~\reftab{comparison} for a comparison with prior works).
Additionally, by expanding well-established constructs, the algorithms remain simple and understandable.
Furthermore, the two algorithms are thrifty in the amount of metadata they use.
Regarding standard approaches in the field, they solely need a set to store non-conflicting messages and a scalar.

\input{comparison}

\paragraph{Outline}
The remaining of this work is organized as follows.
\refsect{related-work} reviews existing literature on atomic multicast.
\refsect{problem} presents the system model and defines generic multicast as well as some related properties of interest.
\refsect{nofault} depicts our base (non-fault-tolerant) solution.
\refsect{gmcast} extends the base solution into a fully-fledged fault-tolerant algorithm.
Additionally, this section argues about the performance of the algorithm and its correctness.
We close in \refsect{conclusion} with a summary of the paper and a prospect of future works.

%%%%%%%%%%%%%%%%%%%%%%%%%%%%
%
% HISTORY AND ALGORITHMS
%
%%%%%%%%%%%%%%%%%%%%%%%%%%%%
\section{Background}
\labsect{related-work}

The literature on atomic multicast is rich, finding its roots in the early works on group communication primitives.
In what follows, we offer a brief tour of the topic, underlining the key ideas and algorithmic principles.
\reftab{comparison} lists the most recent solutions and compares them against the two algorithms we cover in \refsecttwo{nofault}{gmcast}.
\reffigure{algorithm-tree} gives the genealogy of the two algorithms.

\paragraph{Early Solutions}
The first solution to atomic multicast appears in the work of \textcite{birman1987reliable}.
In this work, the authors describe an unpublished algorithm by Dale Skeen that use priorities to order the delivery of messages.
The algorithm employs a two-phase approach.
In the first phase, the sender disseminates the message to its destination group and awaits their responses.
Upon receiving such a message, an algorithm assigns it a priority greater than prior messages and sends this priority back to the sender.
The sender collects responses from the destination groups.
It computes the highest priority among all the proposals and informs the destination group.
Each process in the destination group assigns to the message this new priority.
At a process, messages are delivered in the order of their priority, from lowest to highest.

In modern terms, the solution of Skeen is a \emph{timestamping} algorithm.
This early solution does not tolerate failures:
when the sender fails, the algorithm stalls.
\textcite{birman1987reliable} propose that another process takes over the role of the sender when it fails.
This requires a synchronous system.
However, in general, distributed systems are partially synchronous.
In this context, one cannot detect accurately if a process in the destination group has failed.
To deal with it, \textcite{GuerraouiS97} propose to partition the recipients of a message into (disjoint) consensus groups.
Each group maintains a clock that it uses and advances to propose timestamps.
A similar solution is proposed by \textcite{fritzke1998fault} and \textcite{delporte00}.
From a high-level perspective, all these approaches can be seen as replicating the logic of a Skeen process over a reliable group of machines.

\paragraph{Message Semantics}
In Skeen's approach, each process maintains a logical clock it uses to assign timestamps.
Once the final timestamp of a message is known, the logical clock is bumped to a higher value.
This ensures progress as, otherwise, the message would stall until enough (prior) messages get timestamped.
In this schema, right before the clock is bumped, another message can sneak in and retrieve an earlier timestamp.
This delays delivery.
Such a situation is due to contention, creating a convoy effect in the system \cite{ahmed2016convoy}.
To deal with it, \textcite{ahmed2016convoy} propose to leverage the semantics of messages.
If from the application's perspective the two messages commute, they do not need to wait for each other.
Such an idea finds its root in the generalizations of atomic broadcast \cite{pedone1999generic,lamport2005generalized}.
These algorithms use the semantics of messages to deliver them earlier.
Building upon this, \textcite{ahmed2016convoy} define generic multicast as a variation of atomic multicast that understands the message semantics.
However, the definition in \cite{ahmed2016convoy} only applies to the failure-free case.
We extend it to failure-prone systems in \refsect{problem}.
Notice that the mention of a generalized version of atomic multicast can be traced back in earlier works (\textit{e.g.}, \cite{SchiperPhD}).

\paragraph{Recent Progress}
To bump the clock, \textcite{GuerraouiS97} execute a second agreement per consensus group.
\textcite{schiper2007optimal} observe that this is not always necessary.
In fact, when a message is addressed to a single consensus group, no timestamping mechanism is needed.
\textcite{coelho2017fast} propose the FastCast algorithm.
This solution delivers messages in fewer communication steps by leveraging a speculative path to update the logical clock.
The fast path executes concurrently with a slower (conservative) path.
The message is delivered in four message delays when the fast path correctly guesses the final timestamp.
An improvement is made to the above schema by \textcite{gotsman2019white}.
This algorithm opens the consensus ``black box'' to make additional optimizations.
In the absence of collisions, the solution in \cite{gotsman2019white} delivers a message in just three message delays at the leader of a consensus group.
Non-leader members need an additional communication step.
PrimCast \cite{pacheco2023primcast} is the best collision-free solution known to date.
It cuts the additional message delay by exchanging commit acknowledgment messages across consensus groups (and not just at the leaders as in \cite{gotsman2019white}).
The authors also propose a technique of loosely synchronized logical clocks to reduce the convoy effect.

When multiple messages are addressed concurrently to a process, the above algorithms are slower (``failure-free'' column in \reftab{comparison}).
To avoid this, Tempo \cite{enes2021efficient} leverages the semantics of messages.
Messages are split into $m$ classes (see~\reftab{comparison}).
Two messages conflict, \textit{i.e.}, do not commute, only if they are in the same class (identified with a key).
When messages collide, but do not conflict, Tempo can deliver them in four message delays (``collision-free'' column in \reftab{comparison}).

The work in \cite{sutra2022weakest} characterizes the minimal synchrony assumptions to solve atomic multicast.
In particular, this work shows that, in some cases, weaker assumptions than the standard partitioning into consensus groups is possible.
This paper does not investigate such systems.

\input{algorithm-derivation}

%%%%%%%%%%%%%%%%%%%%%%%%%%%%
%
% SYSTEM MODEL
%
%%%%%%%%%%%%%%%%%%%%%%%%%%%%
\section{Generic Multicast}
\labsect{problem}

This work introduces the generic multicast problem in failure-prone systems.
Generic multicast is a flexible group communication primitive that takes the best of two worlds.
Like atomic multicast, it permits addressing a message to a subset of the processes in the system.
Additionally, the semantics of the messages is taken into account, as in generic broadcast, to order them only where necessary.
Both effects contribute to better delivery performance, as illustrated in \reftab{comparison}.

This section presents the system model and defines the generic multicast problem.
It then illustrates this primitive in the context of distributed storage.
Further, we discuss how to implement generic multicast efficiently.

\subsection{System Model}
\labsect{problem:system}

The distributed system consists of a set $\Pi = \{p_1, p_2, \ldots, p_{n \geq 2}\}$ of processes.
A process may fail by crashing, halting its execution.
When this happens, the process is \emph{faulty}.
Otherwise, it is \emph{correct}.
Processes do not share a common memory and solely rely on passing messages through communication channels to exchange information.
Any two processes are connected with a channel.
A call to $Send\langle m,p\rangle$ sends the message $m$ to the process $p$.
The event $Receive\langle m,p\rangle$ is triggered when message $m$ is received from $p$.
Channels are \emph{reliable}.
This means that if a correct process $p$ sends a message to a correct process $q$, then eventually $q$ receives this message.

The system is partially synchronous.
Partial synchrony is modeled with the help of failure detectors \cite{chandra1996unreliable}.
For pedagogical purposes, we first consider a failure-free system in which all processes are correct (\refsect{nofault}).
Later, we augment the system with failure detectors and quorums (the details are in \refsect{gmcast}).

\subsection{Definition}
\labsect{problem:definition}

Generic multicast permits the dissemination of a message across the system with strong liveness and safety guarantees.
In what follows, we formally define the communication primitive and illustrate its usage in the context of a storage system.

We note $\mathcal{M}$ the set of messages disseminated with generic multicast.
Each such message $m$ carries its source ($m.src \in \Pi$), a unique identifier ($m.id$), as well as the \emph{destination group} ($m.d \subseteq \Pi$).
There are two operations at the interface:
$\gamcast(m)$ multicasts some message $m$ to its destination group, that is, the processes in $m.d$.
$\gdeliver(m)$ triggers at a process when message $m$ is delivered locally.
As common, we consider that processes multicast different messages and that a message is multicast at most once.

In generic multicast, delivery is parameterized with a conflict relation ($\sim$) that captures the semantics of messages from an application perspective \cite{pedone1999generic,lamport2005generalized}.
Relation $\sim$ is irreflexive and symmetric over $\mathcal{M}$.
It is such that if two messages conflict, then they must have a recipient in common:
$m \sim m' \implies m.d \inter m.d' \neq \emptyset$.

For a process $p$ and two conflicting messages $m \sim m'$, we note $m \delOrderOf{p} m'$ when $p$ delivers message $m$ before it delivers message $m'$.
This relation tracks the \emph{local delivery order} at process $p$.
Notice that $m \delOrderOf{p} m'$ holds even if $p$ never delivers message $m'$.
How messages are delivered across the system is captured with the \emph{global delivery order}.
This relation is defined as: $\delOrder = \bigcup_{p\in\Pi} \delOrder_p$.

Generic multicast ensures the following properties on the dissemination of messages:
\begin{itemize}
\item[-] \textbf{Integrity}:
  For some message $m$, every process $p \in m.d$ invokes $\gdeliver(m)$ at most once and only if some process invoked $\gamcast(m)$ before.
\item[-] \textbf{Termination}:
  If a correct process $p$ invokes $\gamcast(m)$ or $\gdeliver(m)$ for some message $m$, eventually every correct process $q \in m.d$ executes $\gdeliver(m)$.
\item[-] \textbf{Ordering}:
  Relation $\delOrder$ is acyclic.
\end{itemize}

Generic multicast establishes a partial order among delivered messages.
In simple terms, it ensures that conflicting messages are delivered in a total order, while non-conflicting ones are delivered in any order.
In both cases, a message is delivered reliably at all the processes in its destination group.

\paragraph{Flexibility}
Generic multicast is a flexible communication primitive.
If the conflict relation binds any two messages ($\sim=\mathcal{M} \times \mathcal{M}$), we obtain atomic multicast.
Conversely, if there is no conflict ($\sim=\emptyset$), the above definition is the one of reliable multicast.
These observations also extend naturally to the case where messages are addressed to everybody in the system.
In such a case, the specification boils down to the one of generic broadcast \cite{pedone1999generic}.
Then, if all messages conflict, we obtain atomic broadcast, and when there are no conflicts, the specification of reliable broadcast.

From what precedes, by adjusting the conflict relation and/or the destination groups, we can adapt the behavior of generic multicast to suit an application's needs.
This permits tailoring an implementation to a specific context, without modifying it.
In the following section, we illustrate such an idea in the context of a storage system.

\subsection{Usage: A Key-Value Store}
\labsect{problem:illustration}

To illustrate the above definitions, let us consider a key-value store.
Such a service is common in cloud infrastructures where it can be accessed concurrently by multiple remote clients (typically, application backends).
A key-value store maps a set of keys to a set of values.
In detail, its interface consists of three operations:
a call to \aread(k) reads the content stored under key $k$,
operation \awrite(k,v) maps key $k$ to value $v$, and
\acas(k,u,v) executes a compare-and-swap, that is it stores $v$ under key $k$, provided the associated value was previously $u$.
Additionally, the interface includes a batch operator $\abegin{\ldots}\aend$ which permits grouping several of these operations.

Modern key-value stores replicate data across several availability zones and/or geographical regions.
This improves fault tolerance and data locality.
A standard approach to implementing such systems is to rely on the use of atomic broadcast in conjunction with atomic commitment.
For instance, this is the design of Spanner \cite{corbett2013spanner}.
In Spanner, each unit of replication, or \emph{tablet}, is replicated across multiple data centers.
Operations on a tablet are ordered with the Paxos consensus algorithm.
When a batch of operations executes over multiple tablets, the replicas run two-phase commit (2PC) \cite{gray1978notes, lampson1979crash} to agree on committing or aborting its changes.

Atomic multicast provides an alternative design \cite{granola,BenzMPG14a}.
Upon executing an operation, a message is multicast to the replicas holding a copy of the corresponding data items.
When the message is delivered, replicas apply the operation locally.
(If the message is a batch, an additional phase is needed, as detailed in \cite{BezerraPR14}.)
Thanks to the ordering property of atomic multicast, this approach also provides strong consistency to the service clients.

In this context, generic multicast can be used as a drop-in replacement to atomic multicast.
For this, we need to define the conflict relation ($\sim$) appropriately.
Given an operation $c$, $\isRead(c)$ evaluates to true when $c$ is a read operation.
We write $\key(c)$ the key accessed by $c$.
Operations $c$ and $d$ conflict when they access the same key and one of them is an update.
Formally,
\begin{equation}
  \labequation{rw}
  c~\sim~d
  ~\equaldef~
  \key(c) = \key(d) ~\land \neg~(\isRead(c) \land \isRead(d))
\end{equation}
For some message $m$, we write $\operations(m)$ the operations $m$ disseminates.
Given two distinct messages $m$ and $m'$, the conflict relation is then defined as:
\begin{equation}
  \labequation{conflict}
  m~\sim~m'
  ~\equaldef~
  \exists c,d \in \operations(m) \times \operations(m') \ldotp c \sim d 
\end{equation}

With generic multicast, operations on the key-value store that do not access the same items commute.
Hence, this group communication primitive permits to order messages only where needed;
this is more flexible and potentially also faster.
If the workload is contended on some specific keys, the system coordinates only the access to those hot items.
If later, the workload becomes uniform, there is no need to re-visit the architecture, as the same guarantees still hold.

In the remainder of the paper, we explain how to implement generic multicast efficiently.
Before detailing these solutions, we first discuss the properties of interest they should have.

\subsection{Implementation Properties}
\labsect{problem:properties}

Atomic multicast satisfies all the requirements of generic multicast.
Unfortunately, it also orders more messages than necessary.
To avoid this, the authors in \cite{pedone1999generic} introduce the notion of \emph{strictness}.
Such a property captures when an implementation permits non-conflicting messages to be delivered in different orders.
The definition given in \cite{pedone1999generic} is for generic broadcast.
It is extended naturally to our context as follows:

\begin{itemize}
\item[] (\emph{Strictness})
  Consider two processes $p$ and $q$ and two non-conflicting messages $m$ and $m'$ with $\{p, q\} \subseteq m.d \cap m'.d$.
  There exist an execution $\epsilon$ during which where $p$ delivers $m$ before $m'$, while $q$ delivers $m'$ before $m$.
\end{itemize}

Another property of interest is with respect to the destination group.
Indeed, it is possible to implement generic multicast by delivering messages first everywhere, with atomic (or generic) broadcast, and then filtering them out based on their destinations.
However, this defeats the purpose of targeting a subset of the system.
In \cite{guerraoui2001genuine}, the authors introduce a minimality property that rules out such approaches.

\begin{itemize}
\item[] (\emph{Minimality})
  In every run $\epsilon$, if some correct process $p$ sends or receives a (non-null) message in $\epsilon$, there exists a message $m$ multicast in $\epsilon$ with $p \in m.d$.
\end{itemize}

The two properties above are of interest from an implementation perspective.
They ensure the communication primitive is flexible enough and appropriately leverages the semantics of messages it disseminates.
In what follows, we complement these with metrics that capture performance.

\subsection{Time Complexity}
\labsect{problem:latency}

In \cite{gotsman2019white}, the authors define the notions of failure-free and collision-free latency.
These metrics are defined when the system is stable, that is, when there are no failures and the system behaves synchronously.
We now extend such metrics to the context of generic multicast.
At core, we take into account the semantics of messages.
Runs are classified whether they contain concurrent \emph{and conflicting} messages, or not.

The \emph{delivery latency} of a message is the time interval between the moment it is multicast and delivered everywhere.%
\footnote{
In \cite{gotsman2019white}, the authors consider the moment the message is first delivered.
Here, we follow the definition proposed in \cite{pacheco2023primcast}.
}
The \emph{conflict-free latency} is the maximum delivery latency for a message when there are no conflicting messages.
A message $m$ precedes a message $m'$ if $m$ is last delivered before $m'$ is multicast.
\footnote{
The term ``before'' refers here to real-time ordering.
}
Two messages $m$ and $m'$ are concurrent when neither $m$ precedes $m'$ nor the converse holds.
The \emph{collision-free latency} is the maximum delivery for a message when there is no concurrent message with a common destination.
Last, the \emph{failure-free latency} is the maximum delivery latency for a message in the presence of concurrent and conflicting messages.

To illustrate the above notions, let us go back to our storage example (see \refsect{problem:illustration}).
The failure-free latency defines an upper bound on the time an operation takes to return.
If the operation does not encounter a concurrent operation when accessing the same replicas, it may return earlier;
this is the collision-free latency.
Now in case the operation accesses a cold item (or a batch of them), its response can be even faster:
in this situation, data is identical everywhere, which is similar to running the operation solo from the initial state.
This optimal case is measured with conflict-free latency.

\paragraph{Lower bounds}
As mentioned earlier, atomic broadcast (and thus consensus) trivially reduces to generic multicast.
Hence any known lower bound on this communication primitive applies to generic multicast.
It is well-known that atomic broadcast can deliver at best a message after a single round-trip, that is, two message delays, in a conflict-free (or collision-free) scenario \cite{lamport2005generalized}.
Now, when a collision occurs, three message delays are necessary \cite{Lamport06a}.

\section{A Base Solution}
\labsect{nofault}

This section introduces a base solution to generic multicast.
This algorithm follows the standard schema invented by Skeen \cite{birman1987reliable}:
For a message, each process in its destination group proposes a timestamp.
The highest such proposal is the final timestamp of the message.
Message delivery happens in the order of their final timestamps.

In Skeen's schema, the clock ticks every time a message is received.
Our key observation is that generic multicast needs only this to happen when a conflict occurs.
Moreover, conflict-free messages are delivered in parallel without waiting in contrast to vanilla atomic multicast.
Both mechanisms reduce the latency and the convoy effect in the communication primitive.

\refalg{nofault} depicts an implementation of generic multicast.
This algorithm works when processes are all correct.
Below, we present its variables and the overall logic, then detail the internals.

\input{nofault}

\paragraph{Overview}
At a process, \refalg{nofault} makes use of the following four local variables.
\begin{itemize}
\item[-] $K$:
  A logical clock to propose a timestamp for each message.
\item[-] $Pending$:
  The messages that do not have a timestamp assigned to them so far.
\item[-] $Delivering$:
  The messages whose timestamp is decided and that are ready for delivery.
\item[-] $PreviousMsgs$:
  The messages received since a conflict was detected.
  By construction, this set only contains messages that do not conflict with each other.
\end{itemize}
Using the above variables, \refalg{nofault} proposes and decides timestamps for the messages submitted to generic multicast.
In a nutshell, each newly submitted message is first disseminated to its destinations.
These processes advance their clocks to propose a timestamp for the message.
Such proposals are then gathered, and the highest one defines the final timestamp.
A message is then delivered in the order of its timestamp.

As mentioned earlier, the logical clock advances only when a conflict is detected.
Namely, when a conflict is detected, the set $PreviousMsgs$ is cleared and the clock advances.
This mechanism reduces the convoy effect in the communication primitive.
We further detail this next.

\paragraph{Internals}
To disseminate some message $m$, a process $p$ invokes at \refline{nofault:amcast:1} operation $\gamcast(m)$.
The operation sends a message $\msgBegin(m)$ to all the processes in $m.d$ using the underlying $Send$ communication primitive.
When this happens, we say that $p$ is the \emph{sender} of message $m$.

At \refline{nofault:begin:1}, the handler triggers when process $p$ receives a $\msgBegin(m)$ message from $m$'s sender, $q$.
In such a case, $p$ first checks if $m$ conflicts with a previously received message.
This computation is at \refline{nofault:begin:2}, using the $PreviousMsgs$ variable.
If the process identifies a conflict, it increments its clock and clears the $PreviousMsgs$ set in \reflines{nofault:begin:3}{nofault:begin:4}. 
Then, the process stores the message $m$ and its timestamp proposal, which is the value of its clock.
The process sends to the sender $q$ the proposal for $m$ at \refline{nofault:begin:7} with a \msgPropose message.

When the sender gathers a proposal from all the processes in $m.d$, it computes the final timestamp of $m$ by selecting the highest proposed timestamp.
The sender then disseminates the decision to all the processes with a \msgDeliver message.
Such a computation is in \reflines{nofault:decide:1}{nofault:decide:4}.

Upon receiving the final timestamp $ts_f$ for message $m$, a process checks if its clock is lower than $ts_f$.
If this happens to be the case, the clock is advanced to $ts_f$.
In case a conflict is detected, the clock is also incremented by one and the $PreviousMsgs$ set is cleared.
This step ensures that no earlier message is missed:
when the final timestamp is known, all the messages before $m$ and conflicting with it are stored locally.

Once the above steps are executed, the message is ready to get delivered at \refline{nofault:deliver:1}.
The precondition in this line ensures that delivery happens in the order of the final timestamps.
In case two timestamps are tied between conflicting messages, the identifier of each message provides a deterministic order.
This means that $(m,t) < (m', t')$ reads as $t < t' \lor (t=t' \land m.id < m'.id)$.

\paragraph{Discussion}
Variable $PreviousMsgs$ can grow arbitrarily large when there is no conflict.
This may negatively impact memory usage.
In practice, it suffices to bump the clock and clear the set periodically.
Because processes cannot unilaterally bump their clocks, such a garbage-collection mechanism ought to be deterministic:
it can happen at a given interval or by broadcasting an appropriate command across the system.

\refalg{nofault} does not restrict the conflict relation.
Therefore, it is possible to build the storage example (\refsect{problem:illustration}) utilizing the conflict relation of \refequation{conflict}.
The storage system would be linearizable taking full advantage of non-conflicting messages.

For \refalg{nofault}, we do not present correctness arguments.
They are detailed in the next section where we present a complete fault-tolerant solution.

\section{Dealing With Failures}
\labsect{gmcast}

We now present a fault-tolerant algorithm to solve the generic multicast problem.
This solution merges the logic of \refalg{nofault} with some of the ideas introduced in earlier works (\textit{e.g.}, \cite{fritzke1998fault, schiper2007optimal}).

In what follows, we first detail the necessary additional assumptions on the system model to deal with process failures.
Then, we introduce our solution, argue about its correctness, and prove that its performance matches the results in \reftab{comparison}.

\subsection{Additional Assumptions}
\labsect{gmcast:assumptions}

\paragraph{Fault-tolerance}
To be fault-tolerant, we need to strengthen our computational model. % which is, so far, purely asynchronous.
A standard assumption \cite{defago2004total} is to assume that the system is partitionable into groups.
Internally, each such group is robust enough to solve consensus. 
The destination a message is then defined as the union of one (or more) of these consensus groups.

In detail, we assume a partition $\Gamma=\{g_1,\ldots,g_m\}$ into groups of $\Pi$ ensuring that:
\begin{itemize}
\item[-] \emph{non-empty}:
  $\forall g \in \Gamma \sep g \neq \emptyset$
\item[-] \emph{complete}:
  $\Pi = \bigcup_{g \in \Gamma} g$
\item[-] \emph{disjoint}:
  $\forall g_i,g_j \in \Gamma \sep i \neq j \implies g_i\inter g_j = \emptyset$
\item[-] \emph{covering}:
  $\forall m \in \mathcal{M} \sep \exists \Gr \subseteq \Gamma: m.d = \union_{g \in \Gr} ~g$
\item[-] \emph{reliable}:
  $\forall g \in \Gamma \sep \land$ consensus is solvable in $g $\\
  \hspace*{6.3em} $\land$~ a process in $g$ is correct
\end{itemize}

At the light of the above assumptions, generic broadcast is solvable in every group $g \in \Gamma$.
In what follows, we consider that such an instance exists per group.
For some group $g$, we note $\gbcast{g}(m)$ and $\gbdeliver{g}(m)$ respectively the operation to broadcast and deliver a message $m$ using generic broadcast in group $g$.

\paragraph{Conflict relation}
In what follows, we consider that messages convey a single operation and that they all access the same key.
That is, the conflict relation ($\sim$) abide by \refequation{rw}.
Later in the paper, we discuss this limitation and how it can be lifted.

\subsection{Algorithm}
\labsect{gmcast:algorithm}

Our solution is presented in \refalg{gmcast}.
Below, we provide an overview of the protocol, detail its internals, and later argue about its correctness.

\paragraph{Overview}
\refalg{gmcast} employs the variables listed in \reflines{gamcast:variables:1}{gamcast:variables:4}.
Some of them have the same usage as in \refalg{nofault}.
\begin{itemize}
\item[-] $K$:
  A logical clock.
\item[-] $PreviousMsgs$:
  The messages received since a conflict was detected.
\item[-] $Mem$:
  This map stores the necessary information to process each message.
  It plays the joint roles of $Pending$ and $Delivering$ in \refalg{nofault}.
\end{itemize}

\input{scenario}

\refalg{gmcast} follows the standard transformation invented in \cite{GuerraouiS97}.
Such a transformation makes Skeen's approach fault-tolerant by considering subsets of the destination group that are large enough to agree on a timestamp proposal.
This is the purpose of the assumptions made in \refsect{gmcast:assumptions}.

In detail, for each destination group $m.d$, there exists a partition into a set of groups $\Gr$ such that $m.d = \union_{g \in \Gr}~g$.
In each group $g$, processes agree on a timestamp proposal for message $m$.
Agreement takes place using generic broadcast in group $g$.
Then, as in \refalg{nofault}, the final timestamp $ts_f$ of message $m$ is the highest of such proposals.
To deliver message $m$, it is necessary to know all the messages with a lower final timestamp than $m$.
For this, each group $g$ bumps its clock using generic broadcast a second time.
In case a group has a clock high enough, it simply does nothing.

\reffigure{scenario} illustrates the above logic.
In this failure-free scenario, two groups $g_1$ and $g_2$ partition the destination group of message $m.d$.
Each group computes a timestamp proposal using generic broadcast:
group $g_1$ proposes $42$, while $g_2$ suggests $3$.
This computation occurs for $g_1$ in the \callout[RoyalBlue!20]{blue box}.
In parallel $g_2$ also reaches an agreement on the proposal using the \callout[OliveGreen!20]{green box}.
Then, the two groups exchange their proposal and set the final timestamp of $m$ to $42$.
At group $g_2$, the clock is lower than $42$.
Thus, $g_1$ can immediately deliver the message.
Group $g_2$ has initially proposed $3$.
As a consequence, it needs to bump its clock using generic broadcast a second time (\callout[OliveGreen!20]{green box} at the bottom of \reffigure{scenario}).

\input{gmcast}

\paragraph{Internals}
A message $m$ starts its journey once it is multicast at \refline{gmcast:amcast:1}.
When this happens, the sender $p$ first computes the partition $\Gr$ of the destination group $m.d$.
Then, for each group $g \in \Gr$, $p$ broadcasts a $\beginMsg(m)$ message to $g$.
This disseminates message $m$ to its destination group to compute timestamp proposals.
If $p$ fails in the loop at \refline{gmcast:amcast:3}, a recovery is needed.
We will explain the recovery procedure later in the section.

Upon receiving a $\msgBegin(m)$ message, a process $p$ looks for conflicts in the $PreviousMsgs$ set (\refline{gmcast:begin:2}).
If a conflict exists, $p$ increments its clock and clears $PreviousMsgs$ (\reflinestwo{gmcast:begin:3}{gmcast:begin:4}).
Then, $m$ is added to $PreviousMsgs$ (\refline{gmcast:begin:5}).
These steps are similar to the ones taken in \refalg{nofault}, and they are needed to ensure that conflicting messages end up with different final timestamps.

The timestamp proposal for $m$ is set to the value of the clock $K$.
This proposal is then stored at the local process $p$ in variable $Mem$ (\refline{gmcast:begin:6}).
One can show that the processes in group $g$ compute the exact same proposal for $m$.
As a consequence, there is no need to disseminate the group's proposal internally.
In \reflines{gmcast:begin:7}{gmcast:begin:8}, $p$ sends the proposal to all the other processes in the destination group.
It also sends the proposal to itself to move the message to the decision phase.

Deciding the final timestamp for message $m$ happens in the handler at \reflines{gmcast:decide:1}{gmcast:decide:7}.
The handler has a precondition that requires that a timestamp proposal is known for each group in the partition of the destination group.
The final timestamp of $m$ is stored in variable $ts_f$ and set to the highest such proposal (\refline{gmcast:decide:3}).
If a (consensus) group $g$ is in late, that is, its clock is lower than $ts_f$, it must bump its clock.
This ensures that all the messages lower than $ts_f$ are known locally when $m$ is ready to be delivered (\refline{gmcast:deliver:2}).
Clock synchronization happens in \reflines{gmcast:decide:5}{gmcast:bump:4}.
It broadcasts an $\msgAdvance(ts_f)$ message within the local group.

Within a group $g$, clocks are synchronized using generic broadcast.
This permits processes to make the same timestamp proposal for a message addressed to the group.
More formally, if $p$ and $q$ in $g$ \gbdeliver{g} the same messages, then variable $K$ is in the same state.
For this, we need to carefully define conflicts in the generic broadcast primitive to ensure the group behaves as a unity.
\refequation{gbcast-conflict} specifies how this relation ($\sim_{gb}$) is set to satisfy the requirements for two (distinct) messages $m$ and $m'$.
\begin{equation}
  \labequation{gbcast-conflict}
  \begin{array}{l@{~}l@{~}l}
    m~\sim_{gb}~m' \equaldef & \lor & m = \msgAdvance(\any) \\
    & \lor & m' = \msgAdvance(\any) \\
    & \lor & m = \msgBegin(x) \land m' = \msgBegin(y) \land x~\sim~y
  \end{array}
\end{equation}

Message $m$ ends its journey when the handler at \refline{gmcast:deliver:1} triggers.
This happens when the clock is higher than $m$'s final timestamp.
As in \refalg{nofault}, such a precondition ensures that every message with a lower final timestamp is known locally.
(This precondition can be added to \refalg{nofault} but it is vacuously true.)
The rest of the handler is identical to the failure-free case, that is message delivery occurs in the timestamp order.

\subsection{Failure Recovery}
\labsect{gmcast:recovery}

\input{recovery}

\refalg{gmcast} is always safe.
However, the algorithm would block when a process fails without a recovery procedure.
In detail, if the sender crashes before it finishes the loop at \refline{gmcast:amcast:3}, some processes in the destination group may never compute a timestamp proposal.
To avoid this, we need a recovery mechanism.
This mechanism is detailed in \refalg{recovery}.

\refalg{recovery} makes use of an unreliable failure detector, denoted $\mathcal{D}$.
This failure detector returns a list processes that are suspected to have failed \cite{chandra1996unreliable}.
Failure detection $\mathcal{D}$ must ensure that a faulty process cannot remain unsuspected forever.
Notice that this does not require any assumption on the underlying system.
In particular, returning $\Pi$ is a valid implementation of failure detector $\mathcal{D}$.
Better failure detection is interesting for performance, but does not impact correctness nor liveness.

At a process $p$, the recovery mechanism in \refalg{recovery} works as follows:
It triggers when for some message $m$, $m$ is pending ($\msgPropose(m, \any) \in Mem$) and its sender is suspected ($m.src \in \mathcal{D}$).
In such a case, for every group $g$ partitioning the destination, if $p$ did not receive a timestamp proposal from $g$, $p$ broadcasts a $\msgBegin(m)$ message to $g$ (\reflines{gmcast:recovery:2}{gmcast:recovery:3}).
When process $p$ is correct, the missing groups deliver a $\msgBegin(m)$ message at \refline{gmcast:begin:1}, and then send their timestamp proposals.
Otherwise, recovery is attempted again by another process, until it succeeds eventually (which happens eventually from the assumptions in \refsect{gmcast:assumptions}).

\subsection{Correctness}
\labsect{gmcast:correctness}

In what follows, we sketch the correctness of \refalg{gmcast}.
\iflong
A complete proof is available in the appendix.
\else
A complete proof appears in the long version of this paper \cite{long-version}.
\fi
\refalg{gmcast} was also verified using \tlaplus.
The specification can be found online \cite{tla-specification}.

\paragraph{Safety}
Integrity follows from the pseudo-code of \refalg{gmcast}.
In detail, a message $\msgBegin(m)$ is delivered from generic broadcast at most once (\refline{gmcast:begin:1}).
Message $m$ starts as a $\msgPropose(m, \any)$ message in variable $Mem$, then it is transformed into a $\msgDeliver(m, \any)$ message.
Once $m$ escapes variable $Mem$, it is delivered at the local process and never re-appears again.

Ordering is proved using the following safety invariant (\SAF):
For any two conflicting messages $m$ and $m'$ delivered in the run with respectively timestamps $t$ and $t'$, if $m \delOrderOf{p} m'$ then $(m,t) < (m',t')$.
To prove invariant (\SAF), two intermediary invariants are obtained.
First, (\SAFa) processes in the destination group agree on the final timestamp of a message.
This result is obtained by using the fact that (consensus) groups agree on the timestamp proposals.
Second, (\SAFb) a process delivers a message only if it has already delivered all the conflicting messages with a lower timestamp.
Such an invariant comes from the fact that the logical clock is monotonically growing, and that it ticks every time a new conflict is detected.

\paragraph{Liveness}
Consider some message $m$ that is either multicast by a correct process, or delivered somewhere.
To prove termination (\LIV), we decompose it into three secondary invariants.
First, (\LIVa) every correct process in $m.d$ eventually stores a $\msgDeliver(m, \any)$ message in variable $Mem$.
Once this holds, invariant (\LIVb) establishes that the clock of a correct process is eventually synchronized to $m$'s final timestamp.
Invariant (\LIVc) states that if a correct process stores a $\msgDeliver(m, \any)$ message in the $Mem$ set, then it eventually delivers $m$.
This third invariant is proved with the help of an appropriate potential function.
The conjunction of \LIVa, \LIVb, and \LIVc ensures that \LIV holds.

\subsection{Performance}
\labsect{gmcast:performance}

\refalg{gmcast} ensures the two implementation properties (strictness and minimality) listed in \refsect{problem:properties}.
In what follows, we prove that the algorithm is matching the performance listed in \reftab{comparison}.
Performance is given for non-faulty runs, that is, runs during which there are no failures and no failure suspicions.
In practice, this corresponds to the common case for real systems.
Finally, we conclude the section with a discussion about metadata management.

When measuring latency, we consider that \refalg{gmcast} makes use of the fastest generic broadcast known to date (\textit{e.g.}, \cite{ryabinin2024swiftpaxos,curp,sutra2011fast}).
Such algorithms ensure that in a failure-free run, a message is delivered after two message delays if there are no concurrent conflicting messages, and three otherwise.

Hereafter, we write $R$ the set of failure-free runs for \refalg{gmcast}.
Among these runs, $R_{\mathit{cf}} \subset R$ (respectively, $R_{\mathit{co}} \subset R$) are the conflict-free (resp., contention-free) ones.
Notice that there is no ordering relation between $R_{\mathit{cf}}$ and $R_{\mathit{co}}$.
We examine in order each of these classes to establish the results in \reftab{comparison}.

\paragraph{Conflict-free}
For starters, we consider the class of conflict-free runs.
This class is illustrated in \reffigure{scenario} where $g_1$ delivers the message after just 3 message delays:
The \callout[RoyalBlue!20]{blue box} for generic broadcast takes~2 message delays.
It is followed with the exchange of timestamp proposals, then delivery of the message at $g_1$.

For some run $r$, we note $dl_r(m)$ the delivery latency of a message $m$ in $r$.
    We can establish the following result.

\begin{proposition}
  \labprop{conflict-free}
  In every run $r \in R_{\mathit{cf}}$, $dl_r(m) \leq 3$ for any $m$.
\end{proposition}

\begin{proof}
  After two message delays, each process in $m.d$ delivers $\msgBegin(m)$ at \refline{gmcast:begin:1}.
  There is no conflict in run $r$.
  Hence, variable $K$ equals $0$ and the condition at \refline{gmcast:begin:2} is false.
  After an additional message delay, the timestamp proposals are exchanged at \reflines{gmcast:begin:7}{gmcast:decide:2}. 
  Upon collecting such proposals, the final timestamp $ts_f=0$ is known at \refline{gmcast:decide:4}.
  Hence, $\msgDeliver(m, 0)$ is added to $Mem$.
  No message conflicts with $m$.
  Thus the precondition at \refline{gmcast:deliver:2} is valid and $m$ is delivered after 3 message delays.
\end{proof}

\paragraph{Collision-free}
Recall from \refsect{problem:latency} that a collision-free run is a run in which no two messages are sent concurrently to the same location.
This means that when a message $m$ is multicast, any message $m'$ having a common process $p \in m.d \inter m'.d$ is delivered at that process $p$.
In \cite{gotsman2019white}, the authors introduce the notion of \emph{commit latency} for Skeen-like algorithms.
The commit latency measures the time it takes for a message to get assigned a final timestamp.
For collision-free runs, the commit latency corresponds to the delivery latency (Theorem~3 in \cite{gotsman2019white}).
\refprop{collision-free} establishes that \refalg{gmcast} takes five message delays to commit a message in a collision-free run.
\begin{proposition}
  \labprop{collision-free}
  In every run $r \in R_{\mathit{co}}$, $dl_r(m) \leq 5$ for any $m$.
\end{proposition}

\begin{proof}
  The proof starts similarly to the one of \refprop{conflict-free}.
  At a process, the final timestamp of a message $m$ is known after three message delays:
  Generic broadcast ensures that $\msgBegin(m)$ is delivered after two message delays (because the run is collision-free).
  Then, we add one more message delay due to the exchange of timestamp proposals.
  Let $t$ be the final timestamp of $m$.
  To deliver $m$, the precondition at \refline{gmcast:deliver:2} must be true.
  This precondition demands that the logical clock ($K$) is higher than $t$.
  Once this holds, message $m$ is delivered right away because there are no pending messages locally (as the run is collision-free).
  According to the pseudo-code at \refline{gmcast:decide:5}, either the precondition is true at the time $\msgDeliver(m,t)$ is added to $Mem$ (\refline{gmcast:decide:4}), or an $\msgAdvance(t)$ message is generic broadcast (\refline{gmcast:decide:7}).
  In the latter case, after two more message delays, the code in \reflines{gmcast:bump:1}{gmcast:bump:4} triggers.
  This ensures that the precondition at \refline{gmcast:deliver:2} is true later on.
\end{proof}

To illustrate the above result, let us go back to \reffigure{scenario}.
This figure depicts a collision-free scenario for \refalg{gmcast}.
In particular, because $g_2$ has a clock smaller than the decided timestamp ($42$), it needs to bump its clock.
This computation happens using generic broadcast (second \callout[OliveGreen!20]{green box} in \reffigure{scenario}) and takes two more message delays.
Once the clock is bumped (bottom of the figure), the message is delivered.
In total, it takes $g_2$ five message delays to deliver the message in this scenario.

\paragraph{Failure-free}
The failure-free latency is defined as the clock update latency plus the commit latency in failure-free runs \cite{gotsman2019white}.
\reflem{failure-free:commit} proves that the commit latency of a message is seven for this class of runs.

\begin{lemma}
  \lablem{failure-free:commit}
  For any $r \in R$, the commit latency in $r$ is 7.
\end{lemma}

\begin{proof}
  The proof is almost identical to the one used in \refprop{collision-free}.
  The sole difference is the time it takes for generic broadcast to deliver a message.
  Here, we count three message delays instead of two.
  This comes from the possible contention in the primitive, leading to an additional message delay in the worst case.
\end{proof}

The \emph{clock update latency} measures the (worst-case) number of message delays to wait before all the messages with a lower timestamp are known at the destination group once the final timestamp is computed.
Equivalently, using the formulation in \cite{pacheco2023primcast}, this is the maximum delay after which no (consensus) group will assign another message a local timestamp smaller than the final timestamp.
For \refalg{gmcast}, \reflem{failure-free:clock-update} tells us that this demands to wait for four message delays.

\begin{lemma}
  \lablem{failure-free:clock-update}
  For any $r \in R$, the clock-update latency in $r$ is 4.
\end{lemma}

\begin{proof}
  Consider that some message $m$ has a final timestamp $t$.
  If a message $m'$ ends up with a lower timestamp than $t$, this message is delivered right before the clock is updated to (at least) $t$.
  In the worst case, $m'$ is generic broadcast at this point by another (consensus) group, say $g$.
  It takes three message delays to get delivered at $g$.
  Then, the timestamp proposals are computed for $m'$ and exchanged among the group partitioning the destination.
  Hence, four message delays were needed in total to compute the timestamp of $m'$ since the clock was updated.
\end{proof}

In the light of the last two lemmas, we may conclude the following about the performance of \refalg{gmcast} in failure-free runs.

\begin{corollary}
  \labcor{failure-free}
  In every run $r \in R$, $dl_r(m) \leq 11$ for any $m$.
\end{corollary}

%%%%%%%%%%%%%%%%%%%%%%%%%%%%
%
% CONCLUSION
%
%%%%%%%%%%%%%%%%%%%%%%%%%%%%
\section{Closing Remarks}
\labsect{conclusion}

This section discusses our results and some possible extensions before closing.

\paragraph{About genericity}
\refalg{gmcast} is limited to read/write conflicts over a single data item key.
Nevertheless, it is possible to leverage commutativity in other situations.
For instance, the very same algorithm works with two additional classes of operations:
one in which operations conflict with nothing, and another where they conflict with everything.
More complex situations demand replicas to store additional metadata.
To track operations over different keys (or segments of keys), we can de-duplicate variables $K$ and $PreviousMsgs$, using a pair of variables per key.
Notice that \refalg{nofault} does not have such a limitation, and it works for any conflict relation.
We believe that this is due to an inherent trade-off between genericity and fault-tolerance.%
\footnote{
Actually, it is possible to replace variable $PreviousMsgs$ in \refalg{gmcast} with a singleton that simply stores the last operation (a read or a write).
}

Regarding generic group communication primitives, we note that they might not always be the best approach.
For instance, simpler algorithmic solutions can be more efficient in specific scenarios (e.g., RDMA networks~\cite{rdmaPaxos}), or when it could be faster to broadcast messages instead of multicasting them~\cite{schiper2009genuine}.
This comes from the fact that group communication primitives are sensitive to the considered application usage.

\paragraph{Future work}
% real-time
The definition of generic multicast can be extended to capture real-time.
In this case, an ordering relation binds two messages $m$ and $m'$ in case $m$ is delivered at a process before $m'$ is multicast (even if $m$ and $m'$ do not have a common process).
Such a variation makes sense when the communication primitive is used in the context of data replication \cite{BezerraPR14}.
Tempo~\cite{enes2021efficient} ensures such a guarantee.
Without it, a message can get delivered earlier in the algorithm after 3 message delays.
We could extend \refalg{gmcast} to satisfy this additional property.
In \refalg{gmcast}, a process must know that the clock of its consensus group is after its final timestamp to deliver it.
It suffices to change the precondition into the clocks of all the consensus groups in the destination.

% conflict classes
In \reftab{comparison}, the space usage of Tempo is higher than with our algorithms.
This comes from the fact that this algorithm discriminates several classes of conflicts (typically using a key per conflict class).
As outlined earlier, one may apply such an idea to \refalg{gmcast} at the cost of storing more metadata at each replica.

% faster delivery
We conjecture that it is possible to cut one message delay in \refalg{gmcast}.
The solution would be similar in spirit to PrimCast~\cite{pacheco2023primcast}:
Each consensus group listens to the decisions taken by the other groups at the destination.
Upon receiving a quorum of commit acknowledgments (2B messages in Generalized Paxos~\cite{lamport2005generalized}), a process knows the clock of a remote group right away.
This can be implemented with the notion of witness as found in state-machine replication algorithms (\textit{e.g.}, \cite{hunt2010zookeeper,curp}).
This optimization skips the need to exchange $\msgPropose$ messages in \refalg{gmcast}.

\paragraph{Conclusion}
This work defines generic multicast in crash-prone distributed systems.
It presents two matching solutions that are variations of the timestamping approach invented by Skeen.
The first solution works in a failure-free environment.
It is extended into a failure-prone algorithm using the standard partitioning assumption over destination groups.
By employing techniques from other well-established works in the literature, the resulting algorithm (and recovery procedure) are relatively simple and understandable.
The algorithm uses a generic broadcast in each group to compute timestamp proposals and deliver messages in a consistent order.
When the run is conflict-free, that is, no two messages conflict, the algorithm delivers each message after three message delays.

%%
%% The acknowledgments section is defined using the "acks" environment
%% (and NOT an unnumbered section). This ensures the proper
%% identification of the section in the article metadata, and the
%% consistent spelling of the heading.
%\begin{acks}
%\end{acks}

%%
%% Print the bibliography
%%
\newpage
\printbibliography
%%
%% If your work has an appendix, this is the place to put it.

\iflong
\clearpage
\appendix
\section{Proofs}
\input{safety}
\input{liveness}
\fi

\end{document}

%% file: comparison.tex
\begin{table}[t]
  \small
  \begin{tabular}{c|c|c|c|c|c}
    \multicolumn{1}{c}{} & \multicolumn{3}{c}{Latency} \\ 
    Algorithm & \makecell{conflict\\-free} & \makecell{collision\\-free} & \makecell{failure\\-free} & generic? & \makecell{metadata\\per msg.} \\ \hline
    FastCast \cite{coelho2017fast} & 8 & 4 & 8 & \NO & $O(1)$ \\ \hline
    White-Box \cite{gotsman2019white} & 5\plusone & 3\plusone & 5\plusone & \NO & $O(1)$ \\ \hline
    PrimCast \cite{pacheco2023primcast} & 5 & 3 & 5 & \NO & $O(1)$ \\ \hline
    \hline
    Tempo \cite{enes2021efficient} & 4 & 5 & 13 & \SOME & $O(m)$ \\ \hline % FIXME check these w. Vitor
    Ahmed et al. \cite{ahmed2016convoy} $\dag$ & 3 & \multicolumn{2}{c|}{5} & \YES & $O(1)$ \\ \hline
    \hline
    \refalg{nofault} $\dag$ & 3 & 3 & 5 & \YES & $O(1)$ \\ \hline
    \refalg{gmcast} & 3 & 5 & 11 & \SOME & $O(1)$ % discounting convoy effect.
  \end{tabular}
  \caption{
    \labtab{comparison}
    Atomic multicast algorithms.
    (\plusone: additional delay to non-leader processes; $\dag$: not fault-tolerant).
  }
\end{table}

%% file: algorithm-derivation.tex
\newcommand{\dashdottedarrow}[1][1.75em]{\mathrel{
    \tikz[baseline,dashdotted]{\draw[->](0,.58ex)--(#1,.58ex)}}}
\tikzstyle{harden} = [->,dashdotted,color=Red]

\newcommand{\dashedarrow}[1][1.75em]{\mathrel{
  \tikz[baseline]{\draw[dash pattern=on .25em off .1 em,->](0,.58ex)--(#1,.58ex)}}}
\tikzstyle{generalize} = [->,color=OliveGreen,dash pattern=on .25em off .1 em]

\newcommand{\arrow}[1][1.75em]{\mathrel{
    \tikz[baseline]{\draw[->](0,.58ex)--(#1,.58ex)}}}
\tikzstyle{optimize} = [->,color=RoyalBlue]

\begin{figure}[!t]
  \centering
  \small
    \begin{tikzpicture}[
        node distance={45mm},
        resize/.style = {scale=0.8},
        every node/.style={text width=6em, align=center}
    ]
        \node (1) {Skeen \cite{birman1987reliable}}; 
        \node (2) [right of=1] {\textcite{ahmed2016convoy}}; 
        \node (3) [below=6mm of 2] {\refalg{nofault} (\refsect{nofault})};
        \node (4) [below=23mm of 1] {\textcite{fritzke1998fault}}; 
        \node (5) [below=6mm of 4] {\textcite{schiper2007optimal}};
        \node (6) [below=21.3mm of 3] {\refalg{gmcast} (\refsect{gmcast})};

        \node at (0, 1) {\textbf{Atomic Multicast}};
        \node at (4.5, 1) {\textbf{Generic Multicast}};
        \node[rotate=90] at (-1.5, -.7) {\textbf{No faults}};
        \node[rotate=90] at (-1.5, -3.4) {\textbf{Crash-Stop}};

        \draw[generalize] (1) --(2) ; 
        \draw[harden] (1) -- (4);
        \draw[optimize] (2) -- (3); 
        \draw[harden] (3) -- (6); 
        \draw[optimize] (4) -- (5);
        \draw[generalize] (5) -- (6);

        \draw[gray, dashed, opacity=0.5] (2.35, 1) -- (2.35, -5);
        \draw[gray, dashed, opacity=0.5] (-1.5, -2) -- (5.5, -2);

        \matrix [draw,above] at (2.2,1.7) {
          \node[text width=23em] (a0) {
            {\color{Red}$\dashdottedarrow$} harden \hspace{2em}
            {\color{OliveGreen}$\dashedarrow$} generalize \hspace{2em}
            {\color{RoyalBlue}$\arrow$} optimize};\\
          %% \node [resize, anchor=west] {1. Solve generic multicast.}; \\
          %% \node [resize, anchor=west] {2. Advance generic multicast.}; \\
          %% \node [resize, anchor=west] {3. Replication for fault-tolerance.}; \\
          %% \node [resize, anchor=west] {*. Optimizations.};\\
        };
    \end{tikzpicture}
    \caption{
      \labfigure{algorithm-tree}
      Structured derivations of generic multicast.      
    }
\end{figure}

%% file: nofault.tex
\begin{algorithm}[t]
    \caption{Base generic multicast -- code at $p$.} 
    \labalg{nofault}
    % Removes the `procedure` prefix, we write what we want.
    % Add the -1mm to remove the automatic padding added with an empty string.
    \algrenewcommand\algorithmicprocedure{\hspace{-1mm}}
    \begin{algorithmic}[1]
        \State \textbf{Variables:}
        \State $K \gets 0$ 
        \State $Pending \gets \emptyset$
        \State $Delivering \gets \emptyset$
        \State $PreviousMsgs \gets \emptyset$
        \vspace{.5em}
        
        \Procedure{\textbf{procedure} \gamcast}{m} \labline{nofault:amcast:1}
            \ForAll{$q \in m.d$} \labline{nofault:amcast:2}
                \State $Send\langle \beginMsg(m), p\rangle$ \labline{nofault:amcast:3}
            \EndFor
        \EndProcedure
        \vspace{.5em}
        
        \Procedure{\textbf{when} $Receive \langle \beginMsg(m),q\rangle$}{} \labline{nofault:begin:1}
            \If{$\exists\; m' \in PreviousMsgs \sep m \sim m'$} \labline{nofault:begin:2}
                \State $K \gets K + 1$ \labline{nofault:begin:3}
                \State $PreviousMsgs \gets \emptyset$ \labline{nofault:begin:4}
            \EndIf
            \State $PreviousMsgs \gets PreviousMsgs \cup \{m\}$ \labline{nofault:begin:5}
            \State $Pending \gets Pending \cup \{(m, K)\}$ \labline{nofault:begin:6}
            \State $Send\langle\texttt{Propose}(m, K), q\rangle$ \labline{nofault:begin:7}
        \EndProcedure
        \vspace{.5em}
        
        \Procedure{\textbf{when} $\forall q \in m.d \sep Receive\langle\texttt{Propose}(m, ts), q\rangle$}{} \labline{nofault:decide:1}
            \State $ts_f \gets$ max($\{ts ~:~ Receive\langle \texttt{Propose}(m, ts), q \rangle\}$) \labline{nofault:decide:2}
            \ForAll{$q \in m.d$} \labline{nofault:decide:3}
                \State $Send\langle \deliverMsg(m, ts_f), q\rangle$ \labline{nofault:decide:4}
            \EndFor
        \EndProcedure
        \vspace{.5em}
        
        \Procedure{\textbf{when} $Receive\langle \deliverMsg(m, ts_f), p\rangle$}{} \labline{nofault:advance:1}
            \If{$ts_f > K$} \labline{nofault:advance:2}
                \If{$\exists m' \in PreviousMsgs \sep m \sim m'$} \labline{nofault:advance:3}
                    \State $K \gets ts_f + 1$ \labline{nofault:advance:4}
                    \State $PreviousMsgs \gets \emptyset$ \labline{nofault:advance:5}
                \Else
                    \State $K \gets ts_f$ \labline{nofault:advance:6}
                \EndIf
            \EndIf
            \State $Pending \gets Pending \setminus \{( m,\_)\}$ \labline{nofault:advance:7}
            \State $Delivering \gets Delivering \cup \{(m, ts_f)\}$ \labline{nofault:advance:8}
        \EndProcedure
        \vspace{.5em}

        \Procedure{\textbf{procedure} \gdeliver}{m} \labline{nofault:deliver:1}
            \Statex \textbf{pre:} $\land~ \exists t \sep (m,t) \in Delivering$ 
            \Statex \hspace{1.8em} $\land~ \forall m', t' \sep (m',t') \in Delivering \cup Pending$
            \Statex \hspace{6em} $\implies (m \not\sim m' \vee (m,t) < (m',t'))$ \labline{nofault:deliver:2}
            \State $Delivering \gets Delivering \; \setminus \; \{(m,t)\}$ \labline{nofault:deliver:3}
        \EndProcedure
    \end{algorithmic}
\end{algorithm}

%% file: scenario.tex
\begin{figure}[t]
  \centering
  \small
    \begin{tikzpicture}[>=latex,scale=.75]

        \coordinate (A) at (1.3,7);
        \coordinate (B) at (1.3,-.5);
        
        \coordinate (C) at (4,7);
        \coordinate (D) at (4,-.5);
        \coordinate (E) at (5,7);
        \coordinate (F) at (5,-.5);
        
        \coordinate (G) at (6.5,7);
        \coordinate (H) at (6.5,-.5);
        \coordinate (I) at (7.5,7);
        \coordinate (J) at (7.5,-.5);
        
        \draw (A)--(B) (C)--(D) (E)--(F) (G)--(H) (I)--(J);
        \draw (A) node[above]{sender};
        
        \draw[decorate,decoration={brace,raise=1ex}]
            (3.85,7) -- (5.15,7) node[midway,yshift=1.5em]{$g_1$};
        
        \draw[decorate,decoration={brace,raise=1ex}]
            (6.35,7) -- (7.65,7) node[midway,yshift=1.5em]{$g_2$};

        \node (1) at (1.3,6.5) [circle, fill, inner sep=1pt] {};
        \node[] () [left = of 1,xshift=3.8em]  {$\gamcast(m)$};

        \node (2) at (1.3,6) [circle, fill, inner sep=1pt]{};
        \node[] () [left = of 2,xshift=3.8em]  {$\gbcast{g_1}$};
        \draw[->] ($(A)!.14!(B)$) to node[above left]{} ($(C)!.14!(D)$);
        \draw[->] ($(A)!.14!(B)$) to node[above left]{} ($(E)!.14!(F)$);

        \node (3) at (1.3,5.5) [circle, fill, inner sep=1pt]{};
        \node[] () [left = of 3,xshift=3.8em]  {$\gbcast{g_2}$};
        \draw[->] ($(A)!.215!(B)$) to node[above left]{} ($(G)!.215!(H)$);
        \draw[->] ($(A)!.215!(B)$) to node[above left]{} ($(I)!.215!(J)$);        

        \filldraw [fill=RoyalBlue!20, draw=black] (3.9,5.85) rectangle (5.1,5.05);
        \filldraw [fill=OliveGreen!20, draw=black] (6.4,5.3) rectangle (7.6,4.5);

        \node (4) at (5,4.3) [circle, fill=RoyalBlue!20, draw=black, inner sep=1pt]{};
        \node (5) at (4,4) [circle, fill=RoyalBlue!20, draw=black, inner sep=1pt]{};
        \node[] () [left = of 5,xshift=3.8em]  {$\gbdeliver{g_1}$};

        \node (6) at (5,3.8) [circle, fill=RoyalBlue!20, draw=black, inner sep=1pt]{};        
        \node (7) at (4,3.5) [circle, fill=RoyalBlue!20, draw=black, inner sep=1pt]{};
        \node[] () [left = of 7, xshift=3.8em]  {$K \gets 42$};

        \node (8) at (5,3.3) [circle, fill=RoyalBlue!20, draw=black, inner sep=1pt]{};
        \draw[->] (8) -- (6.5,2.8);
        \draw[->] (8) -- (7.5,2.46);        
        \node (9) at (4,3) [circle, fill=RoyalBlue!20, draw=black, inner sep=1pt]{};
        \node[] () [left = of 9, xshift=3.8em]  {$Send$};
        \draw[->] (9) -- (6.5,2.22);
        \draw[->] (9) -- (7.5,1.9);

        \node (10) at (6.5,4.1) [circle, fill=OliveGreen!20, draw=black, inner sep=1pt]{};
        \node (11) at (7.5,3.9) [circle, fill=OliveGreen!20, draw=black, inner sep=1pt]{};
        \node[] () [right = of 11,xshift=-3.8em]  {$\gbdeliver{g_2}$};

        \node (12) at (6.5,3.6) [circle, fill=OliveGreen!20, draw=black, inner sep=1pt]{};
        \node (13) at (7.5,3.4) [circle, fill=OliveGreen!20, draw=black, inner sep=1pt]{};
        \node[] () [right = of 13,xshift=-3.8em]  {$K \gets 3$};

        \node (14) at (6.5,3.1) [circle, fill=OliveGreen!20, draw=black, inner sep=1pt]{};
        \draw[->] (14) -- (4,2.2);
        \draw[->] (14) -- (5,2.55);
        \node (15) at (7.5,2.9) [circle, fill=OliveGreen!20, draw=black, inner sep=1pt]{};
        \draw[->] (15) -- (4,1.6);
        \draw[->] (15) -- (5,1.97);
        \node[] () [right = of 15,xshift=-3.8em]  {$Send$};
        
        \filldraw [fill=OliveGreen!20, draw=black] (6.4,1.8) rectangle (7.6,1);

        \node (16) at (4,1.1) [circle, fill=RoyalBlue!20, draw=black, inner sep=1pt]{};
        \node (17) at (5,1.5) [circle, fill=RoyalBlue!20, draw=black, inner sep=1pt]{};
        \node[] () [left = of 16, xshift=3.8em]  {$\gdeliver(m)$};        

        \node (18) at (6.5,.6) [circle, fill=OliveGreen!20, draw=black, inner sep=1pt]{};
        \node (19) at (7.5,.4) [circle, fill=OliveGreen!20, draw=black, inner sep=1pt]{};
        \node[] () [right = of 19,xshift=-3.8em]  {$K \gets 42$};
        
        \node (20) at (6.5,0) [circle, fill=OliveGreen!20, draw=black, inner sep=1pt]{};
        \node (21) at (7.5,-.2) [circle, fill=OliveGreen!20, draw=black, inner sep=1pt]{};
        \node[] () [right = of 21,xshift=-3.8em]  {$\gdeliver(m)$};
    
    \end{tikzpicture}
    \caption{
      \labfigure{scenario}
      A run of \refalg{gmcast}.
    }
\end{figure}

%% file: gmcast.tex
\begin{algorithm}[t]
    \caption{Fault-tolerant generic multicast -- code at $p$} 
    \labalg{gmcast}

    % Removes the `procedure` prefix, we write what we want.
    % Add the -1mm to remove the automatic padding added with an empty string.
    \algrenewcommand\algorithmicprocedure{\hspace{-1mm}}
    \begin{algorithmic}[1]
        \State \textbf{Variables:} \labline{gamcast:variables:1}
        \State $K \gets 0$ \labline{gamcast:variables:2}
        \State $Mem \gets \emptyset$ \labline{gamcast:variables:3}
        \State $PreviousMsgs \gets \emptyset$ \labline{gamcast:variables:4}
        \vspace{.5em}

        \Procedure{\textbf{procedure} \gamcast}{m}\labline{gmcast:amcast:1}
            \State \textbf{let} $\Gr \subseteq \Gamma$ \textbf{such that} $m.d=\union_{g \in \Gr}~g$ \labline{gmcast:amcast:2} \Comment{$\Gr$ is unique} 
            \ForAll{$g \in \mathcal{G}$} \labline{gmcast:amcast:3}
                \State $\gbcast{g}(\msgBegin(m))$ \labline{gmcast:amcast:4}
            \EndFor
        \EndProcedure
        \vspace{.5em}
        
        \Procedure{\textbf{when} \gbdeliver{g}}{\msgBegin(m)} \labline{gmcast:begin:1}
            \If{$\exists~m' \in PreviousMsgs \sep m\sim m'$} \labline{gmcast:begin:2}
                \State $K \gets K + 1$ \labline{gmcast:begin:3}
                \State $PreviousMsgs \gets \emptyset$ \labline{gmcast:begin:4}
            \EndIf
            \State $PreviousMsgs \gets PreviousMsgs \cup \{m\}$ \labline{gmcast:begin:5}
            \State $Mem \gets Mem \union \proposeMsg(m, K)$ \labline{gmcast:begin:6}
            \ForAll{$q \in m.d \setminus g \union \{p\}$} \labline{gmcast:begin:7} 
                \State $Send\langle \proposeMsg(m, K), q \rangle$ \labline{gmcast:begin:8} 
            \EndFor
        \EndProcedure
        \vspace{.5em}

        \Procedure{\textbf{when} $Receive\langle \proposeMsg(m, \any), \any \rangle$}{} \labline{gmcast:decide:1}
            \Statex \textbf{pre:} $m.d = \union~\{g \in \Gamma : \exists q \in g \sep Receive \langle \proposeMsg(m, \any), q \rangle \}$ \labline{gmcast:decide:2}
            \State $ts_f \gets max( \{t ~:~ Receive \langle \proposeMsg(m, t) \rangle \} )$ \labline{gmcast:decide:3}
            \State $Mem \gets Mem \union \deliverMsg(m, ts_f) \setminus \proposeMsg(m, \any)$ \labline{gmcast:decide:4}
            \If{$K < ts_f$} \labline{gmcast:decide:5}
                \State \textbf{let} $g \in \Gamma$ \textbf{such that} $p \in g$ \labline{gmcast:decide:6} \Comment{$g$ is unique}
                \State $\gbcast{g}(\msgAdvance(ts_f))$ \labline{gmcast:decide:7}
            \EndIf
        \EndProcedure
        \vspace{.5em}

        \Procedure{\textbf{when} $\gbdeliver{g}(\msgAdvance(t))$}{} \labline{gmcast:bump:1}
            \If{$t > K$} \labline{gmcast:bump:2}
                \State $K \gets t$ \labline{gmcast:bump:3}
                \State $PreviousMsgs \gets \emptyset$ \labline{gmcast:bump:4}
            \EndIf
        \EndProcedure
        \vspace{.5em}

        \Procedure{\textbf{procedure} \gdeliver}{m} \labline{gmcast:deliver:1}
            \Statex \textbf{pre:} $\land~ \exists t \leq K \sep \deliverMsg(m,t) \in Mem$ 
            \Statex \hspace{1.8em} $\land~ \forall m',t' \sep \proposeMsg|\deliverMsg(m',t') \in Mem$
            \Statex \hspace{6em} $\implies (m \not\sim m' \vee (m,t) < (m',t'))$ \labline{gmcast:deliver:2}
            \State $Mem \gets Mem \setminus \deliverMsg(m, t)$ \labline{gmcast:deliver:3}
        \EndProcedure
    \end{algorithmic}
\end{algorithm}

%% file: recovery.tex
\begin{algorithm}[t]
    \caption{Recovery mechanism}
    \labalg{recovery}

    \begin{algorithmic}[1]
      \Procedure{recover}{$m$}
      \Statex \textbf{pre:} $\land~ \exists ~m \sep \proposeMsg(m,\any) \in Mem$ 
      \Statex \hspace{1.8em} $\land~ m.src \in \mathcal{D}$ \labline{gmcast:recovery:1}
      \ForAll{$g \in \Gamma : g \subseteq m.d$} \labline{gmcast:recovery:2}
          \If{$\not\exists p \in g \sep Receive\langle \msgPropose(m,\any), p \rangle$} \labline{gmcast:recovery:3}
              \State $GB\text{-}Send_g(\texttt{Begin}(m))$ \labline{gmcast:recovery:4}
          \EndIf
      \EndFor
      \EndProcedure
    \end{algorithmic}
\end{algorithm}

%% file: safety.tex
\subsection{Safety} \labappendix{proofs-safety}

We start the proof of \refalg{gmcast} with some invariants:

\begin{enumerate}[label=\textsf{Base \arabic*}]
\item At a process $p$, relation $\delOrderOf{p}$ is irreflexive. \labinv{deliver-once}
\item The logical clock at a process never decreases. \labinv{never-decrease}
\item If a process sends $\msgPropose(m, t)$ and $\msgPropose(m', t')$ with $m \sim m'$, then $t \neq t'$. \labinv{clock-ticks}
\item Consider a consensus group $g$.
  Processes in $g$ execute the two blocks in lines~\ref{line:alg:gmcast:begin:3}-\ref{line:alg:gmcast:begin:4} and lines~\ref{line:alg:gmcast:decide:6}-\ref{line:alg:gmcast:decide:7} in the same order. \labinv{block-order}
\item Consider a consensus group $g$.
  If $p \in g$ sends a $\msgPropose(m, t)$ and $q \in g$ also sends $\msgPropose(m, t')$, then $t' = t$. \labinv{group-same-proposal}
\end{enumerate}

\begin{enumerate}[label=\textsf{SAF\alph*}]
\item Processes in the destination group of a message $m$ agree on the final timestamp of $m$. \labinv{safa}
\item A process delivers a message only if it has already delivered all the conflicting messages with a lower timestamp. \labinv{safb}
\end{enumerate}

\begin{enumerate}[label=\textsf{SAF}]
\item
  For any two conflicting messages $m$ and $m'$ delivered in a run with respectively timestamps $t$ and $t'$, if $m \delOrderOf{p} m'$ then $(m, t) < (m', t')$. \labinv{saf}
\end{enumerate}

\begin{proof}[Proof \refinv{deliver-once}]
  When $p$ delivers some message $m$, $\deliverMsg(m, t)$ must be stored in $Mem$ (\refline{gmcast:deliver:2}).
  Moreover, when this takes place, $\deliverMsg(m, t)$ is removed from $Mem$ (\refline{gmcast:deliver:3}).
  $\deliverMsg(m, t)$ is added at \refline{gmcast:decide:4}.
  This requires $\msgBegin(m)$ in $Mem$.
  When $p$ executes \refline{gmcast:decide:4}, $\msgBegin(m)$ is removed from $Mem$.  
  Such a message is added when the block in \reflines{gmcast:begin:1}{gmcast:begin:8} executes.
  This block is triggered when $\msgBegin(m)$ is delivered by Generic Broadcast in the (consensus) group of process $p$.
  Because Generic Broadcast delivers each message at most once, this happens at most once.
  From what precedes, the delivery of $m$ occurs at most once which implies that $\delOrderOf{p}$ is irreflexive at $p$.
\end{proof}

\begin{proof}[Proof \refinv{never-decrease}]
    Follows from the pseudo-code of \refalg{gmcast}.
    The clock (variable $K$) is either incremented by one (\refline{gmcast:begin:3}), or bumped to a higher value (\reflines{gmcast:bump:2}{gmcast:bump:3}).
\end{proof}

\begin{proof}[Proof \refinv{clock-ticks}]
  Consider that process $p$ sends messages $\msgPropose(m, t)$ and $\msgPropose(m', t')$ during a run, with $m \sim m'$.
  Such messages are disseminated at \refline{gmcast:begin:8}, after that $\msgBegin(m)$ and $\msgBegin(m')$ were delivered at \refline{gmcast:begin:1} by Generic Broadcast. 
  A message is delivered by Generic Broadcast at most once.
  Without lack of generality, consider that $p$ executes \refline{gmcast:begin:1} for $m$ before $m'$.
  Let $\tau$ and $\tau' > \tau$ be the respective points in time at which this takes place.
  
  Recall that each block executes in full atomically, and thus that they are not interleaved.
  At time $\tau$, process $p$ sends a $\proposeMsg(m, t)$ message at \refline{gmcast:begin:8}.
  Note $K_{\tau}$ the value of the clock variable $K$ on process $p$ at time $\tau$.
  According to the block in \reflines{gmcast:begin:1}{gmcast:begin:8}, $K_{\tau}=t$.
  
  By \refinv{never-decrease}, variable $K$ never decreases.  
  Note $K_{\tau'-1}$ the value of the clock variable on process $p$ right before it executes the delivery of $\msgBegin(m')$ at \refline{gmcast:begin:1}.
  There are two cases to consider.
  (case $K_{\tau'-1} > K_{\tau}$.)
  Using the same reasoning as above, we have $t'  \geq K_{\tau'-1}$.
  (case $K_{\tau'-1} = K_{\tau}$.)
  Because $m' \sim m$, process $p$ must execute \reflines{gmcast:begin:3}{gmcast:begin:4} before sending $\msgPropose(m', t')$.
  It follows that $t' = K_{\tau'-1} + 1$, as required.
\end{proof}

\begin{proof}[Proof \refinv{block-order}]
  The two blocks in lines~\ref{line:alg:gmcast:begin:3}-\ref{line:alg:gmcast:begin:4} and lines~\ref{line:alg:gmcast:decide:6}-\ref{line:alg:gmcast:decide:7} happen when a $\msgBegin$ (or $\msgAdvance$) message is delivered in $g$.
  For some $\msgBegin$ or $\msgAdvance$ message $M$, we note $B_M$ such a block.
  Essentially, the invariant tells us that processes in $g$ executes these blocks in the same sequence.

  The proof is by induction.
  Consider that the property holds until time $\tau-1$.  
  At time $\tau$, assume that a process $p$ executes a block $B_M$.
  Consider the process $q$ that executed the most blocks at time $\tau-1$.
  If $q$ also executed block $B_M$, or in case $q$ is in late wrt. $p$, we are done.
  Otherwise, assume that it executed a block $B_{M'}$, with $M \neq M'$.
  If $M$ or $M'$ is an $\msgAdvance$ message, it is easy to obtain a contradiction according to \refequation{gbcast-conflict} and the ordering property of Generic Broadcast.
  Hence, the sole case to consider is when $M$ and $M'$ are both $\msgBegin$ messages.
  Let $m$ and $m'$ be the two (application) messages they are conveying.
  In case $m$ or $m'$ is a write, because of our additional assumption in \refsect{gmcast:assumptions} and \refequation{rw}, these two messages conflict.
  It follows that $p$ and $q$ received them in different orders despite that they are conflicting according to \refequation{gbcast-conflict};
  a contradiction.
  The only situation that remains to consider is the one in which $m$ and $m'$ are two reads.
  Let $m_0$ and ${m'}_0$ be the two messages in $PreviousMsgs$ that triggered the block in lines~\ref{line:alg:gmcast:begin:3}-\ref{line:alg:gmcast:begin:4}, for respectively $M$ and $M'$ at processes $p$ and $q$.
  Observe that $m_0$ and ${m'}_0$ conflict (they are two writes on the same key).
  Thus, the corresponding $\msgBegin$ messages must have been received in the same order via Generic Broadcast.
  By our induction hypothesis, and a short induction, $p$'s or $q$'s state cannot not sound;
  contradiction.
\end{proof}

\begin{proof}[Proof \refinv{group-same-proposal}]
  This is a straightforward consequence of \refinv{block-order}.
\end{proof}

\begin{proof}[Proof \refinv{safa}]
  Consider a message $m$.
  According to the assumptions in \refsect{gmcast:assumptions}, $m.d$ is partitioned uniquely into a set $\{g_1,\ldots,g_m\}$ of consensus groups.
  A process decides the final timestamp of a message at \refline{gmcast:decide:4}.
  This happens when for each group $g_i$ in the partition, $p$ has a timestamp proposal from $g_i$ (\refline{gmcast:decide:2}).
  Hence, because of \refinv{group-same-proposal}, processes must agree on the same final timestamp.
\end{proof}

\begin{proof}[Proof \refinv{safb}]
  Assume that a process $p$ delivers a message $m$ with a timestamps $t$.
  Let $m'$ be some message conflicting with $m$, also addressed to $p$, whose final timestamp is $t'$ (by \refinv{safa} processes agree on this), with $(m',t') < (m,t)$.
  Below, we establish that $m'$ is already delivered at the time $p$ delivers $m$.

  Consider the point in time $\tau$ where $p$ delivers $m$.
  The block in \reflines{gmcast:deliver:1}{gmcast:deliver:3} is responsible for the delivery.
  In particular, this block has a guard to ensure that things happen in the right order (\refline{gmcast:deliver:2}).
  Namely, $m$ is delivered with timestamp $t$ only if for any message $m'$ in $Mem$ either $m$ does not conflict with $m'$, or $m$ has a lower timestamp than $m'$.  

  If $m'$ is already delivered at time $\tau$, we are done.  
  Otherwise, assume that $m'$ is delivered later, or not delivered at all by $p$.
  Below, we prove that a contradiction is reached.

  Message $m'$ conflicts with $m$ and $(m',t') < (m,t)$.
  Thus at time $\tau$, either $m'$ is in $Mem$ at $p$ with a higher timestamp $t''>t$, or $m'$ is not in $Mem$.
  In the former case, $Mem$ contains a $\msgPropose(m',t'')$.
  In light of the pseudo-code of \refalg{gmcast}, necessarily $t'\geq t'' >t$;
  contradiction.
  Alternatively, consider the case where $m'$ is absent from $Mem$ at $\tau$.
  Let $t''$ be the timestamp assigned at \refline{gmcast:begin:6} by process $p$ to $m'$.
  Necessarily, $t'' \leq t' < t$.
  This event must happen before time $\tau$ due to the precondition in \refline{gmcast:deliver:1}, requiring $K \geq t$.
  It follows that there exists a $\msgPropose(m,t'')$ in $Mem$ at an earlier time than $\tau$.
  Because $m'$ is not delivered yet, this message is still in $Mem$ or by \refinv{group-same-proposal}, there is a $\msgDeliver(m',t')$.
  In both cases, we reach the desired contradiction.
\end{proof}

\begin{proof}[Proof \refinv{saf}]
  The proof follows from \refinv{safa} and \refinv{safb}.
  From \refinv{safa}, the processes in the destination agree on every message's final timestamp.
  \refinv{safb} guarantees delivery with a deterministic agreed order, utilizing the timestamp or the message's identifier to break ties.
  
  With more details, the proof is by contradiction and as follows:
  Assume $m$ and $m'$ are delivered with respectively timestamp $t$ and $t'$ in a run.
  Consider that $(m',t') < (m,t)$ and that for some process $p$, $m \delOrderOf{p} m'$.
  Let $t''$ be the timestamp used by $p$ to deliver $m$.
  Applying \refinv{safa}, $t''=t$.
  By \refinv{safb}, because $(m',t') < (m,t)$ and $p \in m'.d$, $p$ must deliver $(m',t')$ before $m$.
  Hence, $m' \delOrderOf{p} m$ and we have a contradiction by \refinv{deliver-once}.
\end{proof}

\begin{theorem}
    \refalg{gmcast} guarantees the Ordering property of generic multicast.
\end{theorem}

\begin{proof}
  For the sake of contradiction, assume that \refalg{gmcast} violates Ordering.
  By definition, there exists a cycle in the delivery of messages across the system.
  By \refinv{deliver-once}, this cycle contains at least two messages.
  In other words, for some $k \geq 1$, there exist messages $m_0, \ldots, m_{k}$ and processes $q_0, \ldots, q_{k}$ such that $m_0 \delOrderOf{q_0} m_1 \delOrderOf{q_1} \ldots \delOrderOf{q_{k-1}} m_k \delOrderOf{q_k} m_0$.
  For any $i \in [0,k]$, let $t_i$ be the timestamp of $m_i$ when the message is delivered at process $p_i$.
  Applying \refinv{saf}, $(m_i,t_i) < (m_{i+1~[k]},t_{i+1~[k]})$.
  Hence, $(m_0,t_0) < (m_0,t_0)$;
  contradiction.
\end{proof}

%% file: liveness.tex
\subsection{Liveness}\labappendix{proofs-live}

We consider the following invariants:

\begin{enumerate}[label=\textsf{LIV\alph*}]
    \item If a correct process $\gamcast(m)$ or a process\\ $\gbdeliver{}(\msgBegin(m))$, eventually all correct processes in $m.d$ insert a $\msgDeliver(m, \any)$ in $Mem$. \labinv{create-deliver}

    \item If a correct process $p$ adds $\msgDeliver(m, t)$ in the $Mem$ set, eventually, $p$'s clock is equal to or higher than $t$. \labinv{clock-synchronizes}

    \item If a correct process $p$ includes a $\texttt{Deliver}(m, \_)$ in $Mem$, then $p$ eventually $\gdeliver(m)$. \labinv{eventually-deliver}
\end{enumerate}

\begin{proof}[Proof \refinv{create-deliver}]
    The critical lines are the for-loop in \reflinestwo{gmcast:amcast:3}{gmcast:amcast:4}.
    This procedure might have two outcomes.    
    Either every group in $m$'s destination delivers a $\msgBegin(m)$ message, or not.

    For starters, consider the former case.
    At this stage, incorrect process crashes does not affect liveness as we assume each group has at least one correct process.
    In detail, since each group contains one correct process, they share proposals at \reflinestwo{gmcast:begin:7}{gmcast:begin:8}.    
    These processes being correct and the channel reliable, the delivery of $\msgPropose(m, \any)$ messages from each group in $m.d$ eventually takes place at the correct processes in $m.d$.
    Thus, each such process calculates the final timestamp at \refline{gmcast:decide:3}.
    Then, it adds a $\msgDeliver(m, \any)$ message to $Mem$ at \refline{gmcast:decide:4}, which concludes the proof.

    Alternatively, consider the second case, that is some group does not deliver a $\msgBegin(m)$ message.
    We observe that this case cannot happen if the sender is correct.
    Hence, from the assumptions in \refinv{create-deliver}, a process $q$ executes $\gbdeliver{}(\msgBegin(m))$.
    Let $g \subseteq m.d$ be its group.
    Group $g$ contains at least one correct process;
    name it $p$.
    From the Termination property of generic broadcast, $p$ delivers $\msgBegin(m)$ because $q$ did.
    After delivering $\msgBegin(m)$ at \refline{gmcast:begin:1}, $p$ submits its proposal to every other processes in $m.d$ then adds $\msgPropose(m, \any)$ to $Mem$ at \refline{gmcast:begin:6}.
    Since the sender process crashes, process $p$'s failure detector ($\mathcal{D}$) eventually suspects $m.src$, triggering the recovery procedure (\refalg{recovery}).
    Therefore, as $p$ is correct, it successfully broadcasts a $\msgBegin(m)$ message to every other group in $m.d$ (which has not proposed a timestamp to $m$).
    Then, we may close the proof using the first case above.
\end{proof}

\begin{proof}[Proof \refinv{clock-synchronizes}]
  \refinv{create-deliver} guarantees that each correct processes in $m.d$ eventually inserts a $\msgDeliver(m,t)$ message in $Mem$.
    After deciding that $t$ is the final timestamp of $m$, a process verifies if the clock needs synchronization at \refline{gmcast:decide:5}.
    There are two outcomes for this, namely the timestamp is higher or not.
    In the former case, the proof is over.
    In the later, $p$'s local group must synchronize its clock.
    To this end, $p$ broadcasts an $\msgAdvance(t)$ message to the group.
    As $p$ is correct, $\msgAdvance(t)$ is eventually delivered at \refline{gmcast:bump:1}.
    Then, the clock is bumped (if needed) to a higher value than $t$ at \refline{gmcast:bump:3}.
\end{proof}

\begin{proof}[Proof \refinv{eventually-deliver}]
    First of all, let us note $\mathcal{T}$ the global (discret) time of the distributed system.
    We consider the following potential function $\Phi$ in $\mathcal{T} \times \Pi \mapsto 2^{\mathcal{M} \times \mathbb{N}}$:
    given a time $\tau \in \mathcal{T}$ and some process $p \in \Pi$, $\Phi(\tau,p)$ returns all the pairs $(m,t)$ such that message $m$ is stored in $Mem$ at $p$ with timestamp $t$ (i.e., $Mem$ contains a $\msgDeliver(m,t)$ or $\msgPropose(m,t)$ message.)
    Then, for such a set of pairs, $\phi(\tau,p,ts_f)$ are all the messages with a timestamp smaller than $ts_f$.

    Assume a process $p$ inserts $\msgDeliver(m, ts_f)$ in $Mem$.
    Applying \refinv{clock-synchronizes}, the clock of $p$ eventually passes $ts_f$.
    Let $\tau$ be the point in time when this happens.
    By \refinv{safb}, $\phi(\tau',p,ts_f)$ is a decreasing function for any later point $\tau' > \tau$ in time.

    Now, assume that $\phi(\tau,p,ts_f)$ is empty.
    The precondition at \refline{gmcast:deliver:1} for $m$ is true at time $\tau$.
    Indeed, all the messages $m'$ preceeding $m$ would be in $\phi(\tau,p,ts_f)$ which is by assumption empty.
    Moreover, it must be always true at any later point in time.
    Hence, $p$ delivers eventually message $m$.
    
    Otherwise, if $\phi(t,p,ts_f)$ is not empty, we can apply inductively the above reasoning on every message in $\phi(t,p,ts_f)$ starting from its smallest element.
    Thus, from what precedes , function $\phi(t,p,ts_f)$ converges towards an empty set over time and messaage $m$ is eventually delivered.
\end{proof}

\begin{theorem}
    \refalg{gmcast} guarantees the Termination property of generic multicast.
\end{theorem}

\begin{proof}
    Follows from \refinv{create-deliver}, \refinv{clock-synchronizes}, and \refinv{eventually-deliver}.
\end{proof}